\documentclass[12pt]{article}

\usepackage{amsthm, amsmath,amsfonts, amssymb}
\usepackage{fullpage,hyperref}

\usepackage{algorithm}
\usepackage{algorithmic}
\usepackage{xspace}
\usepackage{graphicx}

\theoremstyle{definition}
\newtheorem{theorem}{Theorem}

\newtheorem{lemma}[theorem]{Lemma}

\newtheorem{proposition}[theorem]{Proposition}

\newcommand{\ie}{\emph{i.e.}}
\newcommand{\eg}{\emph{e.g. }}
\newcommand{\Uplus}{U^{+}}
\newcommand{\Uminus}{U^{-}}
\newcommand{\U}{\mathcal{U}}

\renewcommand{\S}{\mathcal{S}}
\newcommand{\F}{\mathbb{F}}
\newcommand{\Fq}{\mathbb{F}_q}

\newcommand{\sand}{\;\land\;}
\newcommand{\Pjt}{P_{j,t}}

\DeclareMathOperator{\polylog}{polylog}
\newcommand{\upStage}{\emph{Up-Stage}\xspace}
\newcommand{\downStage}{\emph{Down-Stage}\xspace}
\newcommand{\recursiveShares}{\emph{RecursiveShares}\xspace}

\begin{document}

\title{Scalable Mechanisms for Rational Secret Sharing}

\author{Varsha Dani \thanks{Department of Computer Science,  
University of New Mexico,  Albuquerque,  NM 87131-1386;
email: {\tt \{varsha, movahedi, saia\}@cs.unm.edu}. 
This research was partially supported by NSF CAREER Award 0644058,
NSF CCR-0313160, and an AFOSR MURI grant.} 
\and Mahnush Movahedi \footnotemark[1] \and Jared Saia \footnotemark[1]}

\date{}

\maketitle
\begin{abstract}

We consider the classical secret sharing problem in the case where all
agents are selfish but rational. In recent work, Kol and Naor show
that, when there are two players, in the non-simultaneous communication model, i.e.\ when rushing is
possible, there is no Nash equilibrium that ensures both players learn
the secret.  However, they describe a mechanism for this problem, for any number of players, that
is an \emph{$\epsilon$-Nash equilibrium}, in that no player can gain more than $\epsilon$
utility by deviating from it. Unfortunately, the Kol and Naor
mechanism, and, to the best of our knowledge, all previous mechanisms
for this problem require each agent to send $O(n)$ messages in
expectation, where $n$ is the number of agents.  This may be
problematic for some applications of rational secret sharing such as
secure multi-party computation and simulation of a mediator.

We address this issue by describing mechanisms for rational
 secret sharing that are designed for large $n$.  Both of our results hold for $n \geq 3$, and are Nash equilbria, rather than just $\epsilon$-Nash equilbria.
 Our first result is a mechanism for $n$-out-of-$n$ rational secret sharing that is \emph{scalable} in the sense that it requires each agent to
send only an expected $O(\log n)$ bits.  Moreover, the latency of this
mechanism is $O(\log n)$ in expectation, compared to $O(n)$ expected
latency for the Kol and Naor result.  Our second result is a mechanism for a
relaxed variant of rational $m$-out-of-$n$ secret sharing where $m =
\Theta(n)$.  It requires each processor to send $O(\log n)$ bits and
has $O(\log n)$ latency.  Both of our mechanisms are non-cryptographic, and
are not susceptible to backwards induction.
\begin{flushright}
\emph{``Three can keep a secret if two of them are dead.''}\\ -
Benjamin Franklin
\end{flushright}
\end{abstract}

\section{Introduction}
Secret sharing is one of the most fundamental problems in security,
and is an important primitive in many cryptographic protocols,
including secure multiparty computation.  Recently, there has been
interest in solving \emph{rational secret
  sharing}~\cite{kol2008games, gordon2006rational, halpern2004rational, 
abraham2006distributed,lysyanskaya2006rationality}.
In this setting, there are $n$ selfish but rational agents, and we
want to distribute shares of a secret to each agent, and design a
protocol for the agents ensures that: (1) if any group of $m$ agents
follow the protocol they will all learn the secret; and (2) knowledge
of less than $m$ of the shares reveals nothing about the secret.
Moreover, we want our protocol to be a \emph{Nash equilibrium} in the
sense that no player can improve their utility by deviating from the
protocol, given that all other players are following the protocol.

Unfortunately, all previous solutions to this problem require each
agent to send $O(n)$ messages in expectation, and so do not scale to
large networks.  Rational secret sharing is a primitive for rational
multiparty computation, which can be used to compute an arbitrary
function in a completely decentralized manner, without a trusted
external party.  A typical application of rational multiparty
computation might be to either run an auction, or to hold a lottery to
assign resources in a network.  It is easy to imagine such
applications where the number of players is large, and where it is
important to have algorithms whose bandwidth and latency costs scale
well with the number of players.  Moreover, in a game theoretic
setting, standard tricks to circumvent scalability issues, like
running the protocol only on a small subset of the players, may be
undesirable since they could lead to increased likelihood of bribery
attacks.

In this paper, we address this issue by designing scalable mechanisms
for rational secret sharing.  Our main result is a protocol for
rational $n$-out-of-$n$ secret sharing that (1) requires each agent to
send only an expected $O(\log n)$ bits; and (2) has $O(\log n)$ expected
latency.  We also design scalable mechanisms for a relaxed variant of
$m$-out-of-$n$ rational secret sharing in the case where $m$ is
$\Theta(n)$. We note however that we pay for these improvements by requiring 
the payers to send $O(\log n)$ rather than a constant number of bits per round.

\subsection{The Problem}

Shares of a secret are to be dealt to $n$ rational but selfish
players, who will later reconstruct the secret from the shares. The
players are \emph{learning-preferring,} in the sense that each player
prefers every outcome in which he learns the secret to any outcome in
which he does not learn the secret.  We note that in some previous
work~\cite{kol2008games,abraham2006distributed} it is further assumed
that the players are \emph{competitive}: they prefer that others do
not learn secret. However, this assumption is used mainly for the
purpose of proving lower bounds, and is omitted in the upper bounds.
We will \emph{not} make this additional assumption.

The secret is an arbitrary element of a large (fixed) finite field
$\Fq$.  At the beginning of the game, a dealer provides the shares to
the players.  The dealer has no further role in the game. The players
must then communicate with each other in order to recover the secret.

Communication between the players is \emph{point-to-point} and through
secure private channels. In other words, if player A sends a message
to player B, then a third player C is not privy to the message that
was sent, or indeed even to the fact of a message having been sent.
Communication is \emph{synchronous} in that there is an upper-bound
known on the maximum amount of time required to send a message from
one player to another.  However, we assume \emph{non-simultaneous}
communication, and thus allow for the possibility of \emph{rushing},
where a player may receive messages from other players in a round
before sending out his own messages.

Our goal is to provide protocols for the dealer and rational players
such that the players following the protocol can reconstruct the
secret.  Moreover, we want a protocol that is \emph{scalable} in the
sense that the amount of communication and the latency of the protocol
should be a slow growing function of the number of players.

\subsection{Related Work}

Since its introduction by Halpern and Teague
in~\cite{halpern2004rational}, there has been significant work on the
problem of rational secret sharing, including results of Halpern and
Teague~\cite{halpern2004rational}, Gordon and
Katz~\cite{gordon2006rational}, Abraham et
al.~\cite{abraham2006distributed}, Lysyanskaya and
Triandopoulos~\cite{lysyanskaya2006rationality} and Kol and
Naor~\cite{kol2008games}.  All of this related work except
for~\cite{kol2008games}, assumes the existence of simultaneous
communication, either by broadcast or private channels.  Several of
the protocols proposed~\cite{gordon2006rational,
  abraham2006distributed, lysyanskaya2006rationality} make use of
cryptographic assumptions and achieve equilibria under the assumption
that the players are computationally bounded.  The protocol
from~\cite{abraham2006distributed} is robust to coalitions; and the
protocol from~\cite{lysyanskaya2006rationality} works in the situation
where players may be either rational or adversarial.

The work of Kol and Naor~\cite{kol2008games} is closest to our own
work.  They show that in the non-simultaneous broadcast model (\ie,
when rushing is possible), there is no Nash equilibrium that ensures
all agents learn the secret, at least for the case of two
players. They thus consider and solve the problem of designing an
$\epsilon$-Nash equilibrium for the problem in this communication
model.  An $\epsilon$-Nash equilibrium is close to an equilibrium in
the sense that no player can gain more than $\epsilon$ utility by
unilaterally deviating from it. Furthermore, the equilibrium they
achieve is \emph{everlasting} in the sense that after any history that
is consistent with all players following the protocol, following the
protocol continues to be an $\epsilon$-Nash equilibrium.  As we have
already discussed, our protocols make use of several clever ideas from
their result.

The impossibility of a Nash equilibrium for two players carries over
to the setting with secure private channels, since there is no
difference between private channels and broadcast channels when there
are only two players. However, one might hope that the algorithm of
Kol and Naor~\cite{kol2008games} could be simulated over secure
private channels to give an everlasting $\epsilon$-Nash
equilibrium. Unfortunately, simulation of broadcast over private
channels is expensive, requiring each player to send $\Theta(n)$
messages per round.

In \cite{RSSpodc11} we overcame this difficulty, providing a
\emph{scalable} algorithm for rational secret sharing, in which each
player only sends $O(1)$ bits per round and the expected number of
rounds is constant (although each round takes $O(\log n)$
time). Moreover, following the protocol is an $\epsilon$-Nash
equilibrium. Unfortunately, a certain bad event with small but
constant probability caused some players, when they recognized it, to
deviate from the protocol so that the equilibrium is not everlasting.
This paper is the full version of \cite{RSSpodc11}.  However, we
improve on the work in \cite{RSSpodc11} in two ways.  First, we remove
all probability of error for $n$-out-of-$n$ secret sharing, and
improve the probability of error for $m$-out-of-$n$ from a constant to
an inverse polynomial.  Second, we show that our new protocol is a
Nash equilibrium, not just an $\epsilon$-Nash equilibrium, as long as
$n \geq 3$.

\subsection{Our Results}

The main result of this paper is presented as Theorem~\ref{thm:main}.
This theorem builds on work from our extended abstract
in~\cite{RSSpodc11}.  It also improves on this result in two ways.
First, our new protocol eliminates the probability of failure when
compared with the protocol in~\cite{RSSpodc11}.  Second, our new
protocol has the added advantage of being a Nash equilibrium, not
merely an $\epsilon$-Nash equilibrium.
 
\begin{theorem}\label{thm:main}
Let $n\ge 3$. There exists a protocol for rational $n$-out-of-$n$
secret sharing with the following properties.
\begin{itemize}
\item The protocol is an everlasting Nash equilibrium in which all
  players learn the secret.
\item The protocol, in expectation, requires each player to send
  $O(\log n)$ bits, and has latency $O(\log n)$.
\end{itemize}
\label{thm:nofn}
\end{theorem}

We also consider the problem of $m$-out-of-$n$ rational secret sharing
for the case where $m<n$.  Designing scalable algorithms for this
problem is challenging because of the tension between reduced
communication, and the need to ensure that \emph{any} active set of
$m$ players can reconstruct the secret.  For example, consider the
case where each player sends $O(\log n)$ messages.  If $m = o(n/ \log
n)$, even if the set of active players is chosen \emph{randomly}, it
is likely that there will be some active player that will never
receive a message from any other active player.  Moreover, even if
$m=\Theta(n)$, if the set of active players is chosen in a worst case
manner, it is easy to see that a small subset of the active players
can easily be isolated so that they never receive messages from the
other active players, and are thus unable to reconstruct the secret.

Despite the difficulty of the problem, scalable rational secret
sharing for the $m$-out-of-$n$ case may still be of interest for
applications like the Vanish peer-to-peer
system~\cite{geambasu2009vanish}.  To determine what might at least be
possible, we consider a significantly relaxed variant of the problem.
In particular, we require $m = \Theta(n)$ and that the set of $m$
active players be chosen independently of the random bits of the
dealer.  In this setting we prove the following.

\begin{theorem}
Let $n \geq 3$.  For any fixed positive $k, \lambda$, and threshold
$\tau$, there exists a protocol for rational secret sharing with
absent players, which with probability at least $1-\frac1{n^k}$ has
the following properties, provided that the subset of $m$ active
players is chosen independently of the random bits of the dealer:
\begin{itemize}
\item The protocol is a Nash equilibrium.
\item The protocol ensures that if at least a $(\tau + \lambda)$
  fraction of the players are active, (i.e. $m/n \ge \tau +
  \lambda$\,) then all active players will learn the secret; and if
  less than a $(\tau - \lambda)$ fraction of the players are active,
  (i.e. $m/n \le \tau -\lambda$\,) then the secret can not be
  recovered
\item The protocol requires each player to send $O(\log^2 n)$ bits,
  and has latency $O(\log n)$
\end{itemize}
\label{thm:mofn}
\end{theorem}

This is an improvement to the $\Theta(n)$-out-of-$n$ result we proved
in \cite{RSSpodc11}, in the sense that the probability of error
in~\cite{RSSpodc11} is a small constant, but here it is $1/n^{k}$.
However, we cannot completely eliminate the probability of failure.

\subsection{Our Approach}

The difficulty in designing a Nash equilibrium in a communication
model where rushing is possible, is that the last player to send out
his share has no incentive to actually do so.  He already has the
shares of all the other players and can recover the secret alone. To
get around this, it is common (see \cite{halpern2004rational,
  abraham2006distributed, gordon2006rational,
  lysyanskaya2006rationality, kol2008games}) for the protocol to have
a number of fake rounds designed to catch cheaters. The uncertainty in
knowing which is the ``definitive'' round, during which the true
secret will be revealed causes players to cooperate.

In the work of Kol and Naor~\cite{kol2008games} this uncertainty is
created by dealing one player only enough data to play until the round
preceding the definitive one.  Thus, there is a single ``short''
player and $n-1$ ``long'' players.  None of the players know whether
they are short or long.  The long players must broadcast their
information every round, since they cannot predict the definitive
round in advance. The short player knows the definitive round in
advance, but has no information about the secret. In the definitive
round the short player is the last to speak so that he (and all the
other players) receives the shares of all the long players and can
recover the secret. His failure to broadcast a message is what cues
the other players to the end of the game, and they too can recover the
secret. Moreover, having learned the secret, the short player cannot
pretend that he actually had a share for that round as the messages
sent by all the players are verified by a tag and hash scheme (see,
e.g., ~\cite{wegman1981new,rabin1989verifiable, kol2008games}). In
fact, it is the small but positive chance of cracking the tag and hash
scheme that results in this being an $\epsilon$-Nash equilibrium
rather than a Nash equilibrium.

Here, we also use short and long players. However we introduce two
novel techniques to ensure scalable communication and to ensure a Nash
equilibrium.  The first technique is to arrange players at the leaves
and nodes of a complete binary tree, and require that the players only
communicate with their neighbors in the tree. The assignment of
players to the leaves is independently random every round, and their
assignment to internal nodes is related to their assignment to leaves
by a labeling of the tree that is common knowledge.  Every round of
the game, information travels up to the root where it is decoded and
then travels back down again to the leaves. The short players are the
parents of the leaves in the definitive round, so that now about half
the players are short players.

The second main idea is that we make use of an iterated secret sharing
scheme over this tree in order to divide up shares of secrets among
the players.  This scheme is similar to that used in recent work by
King and Saia~\cite{king2010breaking} on the problem of scalable
Byzantine agreement, and suggests a deeper connection between the two
problems.

As in previous works~\cite{wegman1981new,rabin1989verifiable,
  kol2008games} we use a tag-and-hash scheme to ensure that players
cannot forge messages in the protocol. We note however, that unlike in
previous work, our use of the verification scheme is such that even by
breaking it, players who have learned the secret cannot prevent other
players from learning it as well. Thus, in our case the small
probability of forging messages without detection does \emph{not}
translate into the protocol being an $\epsilon$-Nash equilibrium.
Instead we show that our protocol is a Nash equilibrium for all the
players.

\subsection{Paper Organization}

The rest of this paper is laid out as follows.  In
Section~\ref{sec:prelim}, we give notation and preliminaries.  In
Section~\ref{s:algnofn}, we describe our algorithm for scalable
$n$-out-of-$n$ secret sharing.  In Section~\ref{sec:analysis}, we
analyze this algorithm; the main result of this section is a proof of
Theorem~\ref{thm:nofn}.  In Section~\ref{s:mofn}, we give our
algorithm and analysis for scalable $m$-out-of-$n$ secret sharing
where $m = \theta(n)$; the main result of this section is a proof of 
Theorem~\ref{thm:mofn}.  Finally in Section~\ref{s:conclusion}, 
we conclude and give directions for future work.

\section{Notation and Preliminaries}\label{sec:prelim}

The secret to be shared is an arbitrary element of a set $\S$. 
There are $n$ players with distinct player IDs in $[n] = \{1, 2, \dots n\}$. 
During the course of the algorithm, we will want to do arithmetic 
manipulations with player IDs and shares of the secret, including adding 
in, or multiplying by random elements to preserve secrecy. In order to be 
able to do these sorts of manipulations, we embed the sets $\S$ 
and $[n]$ into a finite field $\F$ of size $q >\max\{n, |\S|\}$.
The latter embedding will be the canonical one; the former may be arbitrary, 
but is assumed to be known to all parties.

The messages sent by players in the algorithm will be elements of $\F$. 
The length of any such message is $\log q = \Omega(\log n)$. Since our 
goal is to provide a scalable algorithm we cannot afford the message 
lengths to be much bigger than that. We will choose $\F$ to be a prime field
of size $q=O(n)$. We remark that although generally 
$\S$ is of constant size, we can tolerate $|\S|=O(n)$.

\subsection{Utility Functions}
We will denote the utility function of player $j$ by $u_j$. As mentioned 
before, we assume that the players are learning preferring, \ie, each player 
prefers any outcome in which he learns the secret to every outcome in 
which he does not learn the secret. More formally, for outcome $\mathbf{o}$ 
of the game, let $R(\mathbf{o})$ denote the set of players who learn the 
secret. If $\mathbf{o}$ and $\mathbf{o'}$ are outcomes of the game such 
that $j \in R(\mathbf{o}) \setminus R(\mathbf{o'})$, then 
$u_j(\mathbf{o}) > u_j(\mathbf{o'})$. 
As in~\cite{kol2008games}, we denote
\begin{align*}
\Uplus_j &= \max \{ u_j(\mathbf{o})\; | \; j \in R(\mathbf{o}) \}\\
U_j &= \min \{ u_j(\mathbf{o}) \; | \; j \in R(\mathbf{o}) \}\\
\Uminus_j &= \max \{ u_j(\mathbf{o}) \; | \; j \notin R(\mathbf{o}) \}.
\end{align*}
Thus $\Uplus_j$ is the utility to player $j$ of the best possible outcome 
for $j$, $U_j$ is the utility to $j$ of the worst possible outcome in 
which $j$ still learns the secret, and $\Uminus_j$ is the best possible 
utility to $j$ when he does not learn the secret. By the learning-preferring 
assumption, we have for all $j$,
\[
\Uplus_j \ge U_j > \Uminus_j.
\]
We will denote by $\U$, the quantity
\[
\U := \max_{j \in [n]} \frac{\Uplus_j - \Uminus_j}{U_j -\Uminus_j}.
\] 
Note that $\U \ge 1$. We assume that $\U$ is constant, \ie, that it does 
not depend on $n$.\footnote{Technically, we can achieve scalable (polylog)
communication even if we allow $\U$ to be as big as $\polylog(n)$}

We also assume that the utilities are such that \emph{a priori} the players 
have an incentive to play the game rather than just guess the secret at 
random. In other words, we require that 
\begin{equation}\label{eqn:su}
\U < |\S|.
\end{equation}
We have said earlier that $|S|$ may be as big as $n$.  If that is the case, 
then \eqref{eqn:su} is trivially satisfied since $\U =O(1)$. When $\S$ is of 
constant size however, \eqref{eqn:su} is a genuine constraint.

\subsection{Game Theoretic Concepts}
In this section we review some game theoretic solution concepts.

Recall that an $n$-tuple of strategies for an $n$ player game is called 
a \emph{Nash equilibrium} if no player has an incentive to unilaterally deviate 
from the equilibrium strategy, when all others are following it. 


In games of incomplete information which have multiple rounds, there is 
the further question of whether the players are forced to commit to their 
strategies before the start of the game or whether they have the option to 
change strategies in the middle of the game, after some rounds have been 
played and they may learn some new information. Kol and 
Naor~\cite{kol2008games} defined a Nash equilibrium to be 
\emph{everlasting} if after any history that is consistent with all 
players following the equilibrium strategy, it is still true (despite 
whatever new information the players may have learned over that history) 
that a player choosing to deviate unilaterally cannot gain in expectation, 
\ie, following the prescribed strategy remains a Nash equilibrium. 
This is a stronger concept than the usual Nash equilibrium, where 
the strategies are committed to up front.

\section{Algorithm For All Players Present}
\label{s:algnofn}
We now describe our scalable mechanism for $n$-out-of-$n$ secret
sharing. First, in Section~\ref{sec:commtree} we describe the communication 
tree that is used by the dealer and players. An informal description
of the mechanism follows in Section~\ref{sec:alg}. The formal descriptions 
of the dealer's and players' protocols 
appear respectively as Algorithms~\ref{alg:Dealer} and~\ref{alg:Player}.

\subsection{The Communication Tree}\label{sec:commtree}

Recall that a complete binary tree is a binary tree in which all the 
internal nodes have exactly two descendants, all the leaves are at the 
two deepest levels, and the leaves on the deepest level are as far left 
as possible.

Our communication tree is a complete binary tree with $n$ leaves. 
The leaves will be labelled 1 to $n$ from left to right. Next every 
internal node which is a parent of two leaves is labelled with the 
odd label from among its two children. Finally, the remaining internal 
nodes are labelled in order with even numbers, proceeding top to bottom and left to right, starting with 2 at the root. If $n$ is odd, then each even number
appears at some internal node. If $n$ is even, we will place the last 
even number, $n$ at the root, along with 2 (so the root will have two labels.)
The tree thus labelled has the following properties:
\begin{itemize}
\item Every even label occurs at some internal node. (Note that if $n$ is 
odd, there will be an odd label that occurs only at a leaf and not at any 
internal node. This will not matter.)
\item No even labelled internal node has an odd labelled node above it.
\item Every path from root to leaf has exactly one odd label (the same odd 
label may occur once or twice on the path.)
\end{itemize}
Figure~\ref{fig:tree} illustrates the labelling scheme for five and six players.

\begin{figure}
\begin{center}
\includegraphics[scale=0.2]{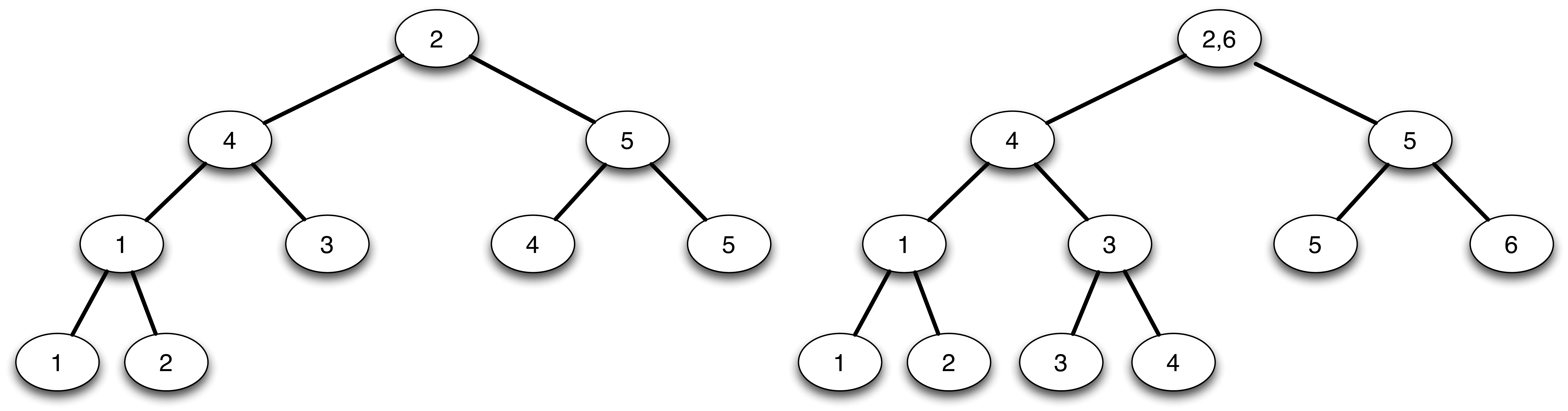}
\end{center}
\caption{Communication trees for five players and six players}
\label{fig:tree}
\end{figure}

\subsection{Our Algorithm}\label{sec:alg}

\begin{algorithm}
\caption{Dealer's Protocol} \label{alg:Dealer} 
$\F$ field of size $q$ (to represent messages in the algorithm)
$n$ players with distinct identifiers in $[n] \hookrightarrow \F$,
$\beta \in (0, 1)$: geometric distribution parameter.
Complete binary tree with $n$ leaves, labelled as described in 
Section~\ref{sec:commtree} known to everyone.
\begin{enumerate}
\item Choose $X, Y,$ independently from a geometric distribution with 
parameter $\beta$. Round $X$ is the definitive one. Short players will receive 
full input for $X-1$ rounds and partial input for round $X$. Long players 
will receive full input for $X+Y-1$ rounds and partial input for round $X+Y$. 
For convenience we will create all the inputs for $X+Y$ rounds, and 
truncate them appropriately before sending them to the players.
\item For each round $t$ between 1 and $L=X+Y$:
\begin{itemize}
\item If $t<X+Y$, choose a random permutation $\pi_t \in S_n$. 
If $t=X+Y$ choose a permutation $\pi_L$ which is random subject to the 
constraint that all the long players (determined by $\pi_X$) are assigned 
to odd labels under $\pi_L$.
For round $t$ player $j$ 
will be assigned to all nodes marked $\pi_t(j)$ in the tree. 
\item If $t=1$, $m_1=0$ // (Otherwise $m_t$ was set in the previous round)
\item For every player $j$, $(\pi_t(j), (\pi^{-1}_t(i)| 
\mbox{node }i \mbox{ is a neighbor of node } \pi_{t}(j) \mbox{ in the tree.}))$,
is a tuple of elements of $\F$ representing $j$'s position and the identities 
of his neighbors in the tree for round $t$. 
$P^{j}_t = (\pi_t(j) + m_t, ((\pi^{-1}_t(i) +m_t| 
\mbox{node }i \mbox{ is a neighbor of node } \pi_{t}(j) \mbox{ in the tree.}))$
is a masked version.
\item Choose a random mask $m_{t+1} \in \F$ (for the {\bf next} round.)
\item Create shares of $m_{t+1}$ by calling 
\recursiveShares$(\mbox{root}, m_{t+1})$.
\item If $t = X$ $s_t \leftarrow $ true secret.\\
 Otherwise, $s_t \leftarrow $ random element of $\S$
\item Create shares of $s_t$ by calling \recursiveShares$(\mbox{root}, s_t)$.
\item Create tags and verification functions for all the messages to be 
sent in round $t$
\item For each player $j$, $j$'s (full) input $I^j_t$ for round $t$  
consists of $P^j_t$, shares of $m_{t+1}$ and $s_t$ corresponding to node 
$\pi_t(j)$, tags to authenticate all messages to be sent by $j$ and 
verification vectors for all the 
messages to be received by $j$. Partial input $\tilde{I}^j_t$ consists of 
all of the above except the authentication tags for sending messages to your 
children (in the down-stage).
\end{itemize}
\item Identify the short players as those players $j$ who are at odd numbered 
nodes in the definitive iteration, \ie, $\pi_X(j)$ is odd.
\item For each short player $j$, send $j$ the list 
$I^j_1,\dots I^j_{X-1}, \tilde{I}^j_X$.
\item For each long player $j$, send $j$ the list 
$I^j_1,\dots I^j_{L-1}, \tilde{I}^j_L$.
\end{enumerate}
\end{algorithm}

\begin{algorithm}
\caption{ \recursiveShares(node $w$ , $\F$-element $y$):}\label{alg:shares}
$n$-leaf complete binary tree global data structure $V$; for node $w'$, 
$V^{w'}$ denotes the location for the data associated with $w'$.\\
Initially called with the root node and the value for which shares are to be 
created, this function populates $V$ with 
intermediate values. The values at the leaves are the shares for the players 
at the corresponding leaves of the  communication tree.
\begin{enumerate}
\item $V^w   \leftarrow y$.
\item If $w$ has children $\ell(w)$ and $r(w)$:
	\begin{enumerate}
	\item Choose random slope $\mu$ from field $\F$.
	\item Let $f$ be the line with slope $\mu$ and y-intercept $y$.
	\item $\recursiveShares(\ell(w), f(-1))$.
	\item $\recursiveShares(r(w), f(1))$.
	\end{enumerate}
\end{enumerate}
\label{alg:recursiveShares} 
\end{algorithm}

\begin{algorithm}
\caption{Create Authentication Data ($\F$-element $y$,):   
// $y$ is the message to be transmitted} \label{alg:verification}
\begin{enumerate}
\item Choose $a \in \F$ and $b \in \F^* = \F\setminus\{0\}$ independently, 
uniformly at random.
\item $c = y + b*a$.
\item $a$ is the tag, to be given to the sender of message $y$, $(b,c)$ is the 
verification vector, to be given to the recipient of the message $y$.  
\end{enumerate}
\end{algorithm}

\begin{algorithm} 
\caption{Protocol for Player $j$} \label{alg:Player} 
S=0; M=0 \\
If at any time you receive spurious messages (messages not expected uder the 
protocol), ignore them. \\
On round $t$:
\begin{itemize}
\item[] \upStage:
\begin{enumerate}
\item $m_t = M$
\item Subtract $m_t$ from all elements of $P^j_t$ to find out your position 
in the tree and the identities of your neighbors for round $t$.
\item (as player at a leaf) Send your shares of $s_t$ and $m_{t+1}$ along with 
their tags to your parent in the tree.
\item (as player at an internal node) 
\begin{enumerate}
\item Receive (intermediate) shares of $s_t$ and $m_{t+1}$ and tags from 
left and right chidren. Use the 
appropriate verification vectors to check that correct messages have been sent.
If a fault is detected (missing or incorrect message) output $S$ and quit.
\item For each of $s_t$ and $m_{t+1}$: interpolate a degree 1 polynomial $f$ 
from $(-1, \mbox{ left-share})$ and $(1, \mbox{ right-share})$. 
Evaluate $f(0)$. This is your share.
\item If you are not at the root, send the above reconstructed shares of $s_t$ 
and $m_{t+1}$ to your parent(s) along with the appropriate tags. If you 
\emph{are} at the root, these shares are the actual values of $s_t$ and 
$m_{t+1}$.
\end{enumerate}
\end{enumerate}
\item[] \downStage:
\begin{enumerate}
\item If you are at the root, set $S=s_t$ and $M=m_{t+1}$ and send these 
values along with authentification tags to your left and right children.
\item Else
\begin{enumerate}
\item (as a non-root internal node) Receive $s_t$ and $m_{t+1}$ and tags from 
your parent and use verificaton vectors to check them. If fault detected, 
output $S$ and quit.
\item Set $S=s_t$ and $M=m_{t+1}$.
\item Send $s_t$ and $m_{t+1}$ to your children along with the appropriate tags.
If you are a short player and have no authentication tags, output $s_t$ and 
quit.
\end{enumerate}
\end{enumerate}
\item[] $t \leftarrow t+1$
\end{itemize}
\end{algorithm}

The dealer is active only once at the beginning of the game, and during 
this phase of the game the players' inputs are prepared.

The dealer independently samples two random variables $X$ and $Y$ from 
a geometric distribution with parameter $\beta$ (to be determined later). 
$X$ will be the definitive iteration, or the round of the game in which 
the true secret is revealed. $Y$ will be the amount of padding on the 
long players' input. Note we have two kinds of players: short players 
will receive enough input to last for $X$ rounds of the game while long 
players will receive enough input to last for $X+Y$ rounds of the game.
The partitioning of players into short and long will be random, and the 
players themselves will not know which are which. This is critical in 
our analysis as is discussed in Section~\ref{sec:analysis}.

Communication between the players in our protocol will be restricted 
to sending messages to their neighbors in the communication tree.  
In order not to reveal which players are the short players, the players 
will be reassigned to new positions in the tree in each round. This is 
accomplished by choosing a random permutation of the players each round 
and assigning them to labelled nodes of the tree according to it.
The short players are the ones who are at odd labelled nodes in the 
definitive round. 

Since the players must be at different nodes in the tree each round, their 
input must contain this information. At the same time, the positions of the 
players for all the rounds cannot be revealed up front, since this may give 
away information about who the short players are. A naive idea to solve this 
problem is, in each round, to distribute shares of the permutation for the 
next round. Then 
during each round, the players could reconstruct the permutation from the 
shares and use it to reposition for the next round. Unfortunately, there 
is a problem with this approach. To represent permutations of $n$ symbols, 
we need a field of size at least $n!$. To transmit elements of such a field,
players would need to send messages of length $\log (n!) \sim n\log n$.
This is unacceptable if we desire scalability. 

To get around this problem, we note that it is not really necessary for 
players to know the entire permutation. Each player only needs to know 
its own position and the identities of its neighbors. We only need a field 
of size order $n$ to encode this, and so, symbols of this field may be 
transmitted with messages of length $\log n$. Since it is dificult 
via share reconstruction to transmit different messages to the leaves of 
the tree, we simply provide each player with a list of positional data for the 
entire game. But in order that players do not know their positional data 
for a round before actually getting to that round, this data is masked 
by adding in a random element of the field. Positional data for the first 
round is sent unmasked. The players also receive iterated shares of the masks
for the next round. Thus, in each round, players reconstruct a mask, and use it 
to unmask the positional data and reposition themselves for the next round.

For each round, the full input consists of the following:
\begin{itemize}
\item iterated shares of a purported secret (the true secret in the 
definitive round);
\item masked versions of positional data for the current round 
(position and identities of neighbors in the tree);
\item shares of the mask for the \emph{next round} of positional data;
\item tags for all the messages to be sent; and
\item verification vectors for all the messages to be received.
\end{itemize}

The iterated shares the players receive are constructed by 
starting with the symbol to be reconstructed at the root and 
recursively constructing 2-out-of-2 Shamir shares down the tree, all 
the way down to the leaves. The shares at the leaves are the iterated 
shares the players receive. See Figure~\ref{fig:shares} and 
Algorithm~\ref{alg:recursiveShares} for details of how the iterated shares are 
constructed. At reconstruction time, shares are sent 
up the tree. At each internal node, a pair of shares received from 
the two children is reconstructed into a degree 1 polynomial which is used 
to obtain the value to be sent further up the tree. At the root, the 
original symbol is reconstructed and transmitted down the tree.  
Note that the advantage of this scheme 
over simply using n-out-of-n Shamir shares is that the size of the 
messages does not increase as the messages are transmitted up the tree.

\begin{figure}
\begin{center}
\includegraphics[scale=0.2]{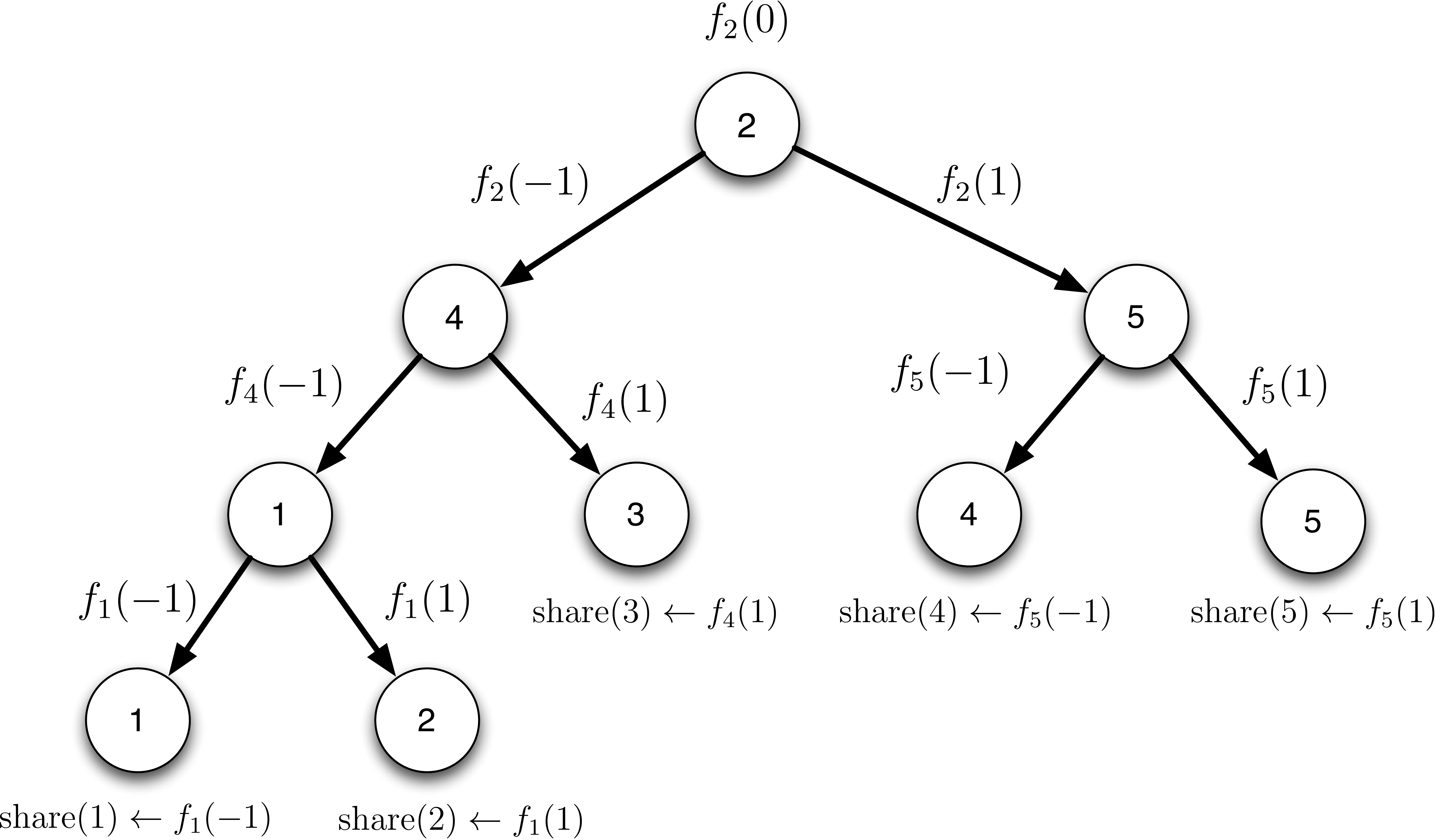}
\end{center}
\caption{Construction of the iterated shares}
\label{fig:shares}
\end{figure}

As mentioned earlier, round $X$ is the definitive round, when the 
encoded symbol is the true secret. Short players receive full input 
for every round prior to this round. For round $X$ they only receive 
partial input. Long players receive input for $X+Y$ rounds. However,
they, too receive only partial input for their last block of input. 
Otherwise, a player would be able to distinguish whether or not he is a 
short player by looking at his last block of input.
Here, partial input consists of all of the pieces of data from the full input, 
\emph{except} the tags to send the decoded message to your children in the 
down stage of the round. 

Since, in the definitive round the short players 
(with odd labels) are in the level above the leaves, and all the long 
players are at internal nodes higher than that in the tree, the long players 
have learned the secret before the short players, although since they have
input for more rounds of the game, they do not know (\ie, cannot guess) that 
it is the definitive round, and that the secret they have learned is in 
fact the  true secret. Thus they send the secret down the tree, and 
eventually it gets to the short players. Thus the short players learn 
the secret as well. Since they have no more input they know that the 
game is over and the secret is the true one. However, since they do not 
have any more authentication data, they cannot gain by remaining in the 
game and trying to fool the others into thinking that the secret has not 
yet been reconstructed. Finally, when the long players do not receive a 
message at the end of the definitive iteration, they too realize the 
game has ended and output the correct secret.

\section{Analysis of Algorithm for All Players Present}\label{sec:analysis}

In this section we will prove Theorem~\ref{thm:main}, which shows that the secret 
sharing scheme we have described is in fact a scalable $n$-out-of-$n$ 
secret sharing scheme, and that it is an everlasting Nash equilibrium 
in which all the players learn the secret. 

We begin by showing that our rescursive scheme for encoding a symbol into $n$ 
iterated shares (Algorithm~\ref{alg:shares}) is an $n$-out-of-$n$ scheme.

\begin{lemma}\label{lem:shares}
Let $\sigma \in \F$ be a symbol that is encoded into $n$ iterated 
shares, $\sigma_1, \dots \sigma_n$, by Algorithm~\ref{alg:shares}. Then $\sigma$
can be decoded from all $n$ of the shares, but knowledge of fewer than 
$n$ of the shares reveals no information about $\sigma$
\end{lemma}
\begin{proof}
That $\sigma$ can be decoded from all $n$ shares follows easily from the fact 
that two points on a line determine it. Since the value at a node is the 
$y$-intercept of a line passing through the points 
$(-1, \mbox{ left-child value})$ and $(1, \mbox{ right-child value})$, it 
can be reconstructed using interpolation. Starting with the shares 
$\sigma_1 \dots \sigma_n$ at the leaves of the tree, we can reconstruct the 
values bottom-up, and the value at the root is $\sigma$ since this is 
exactly the reverse of the process used to create these shares.

To see why fewer than $n$ shares give us no information about $\sigma$, 
observe that the values at the two children of the root were created by 
choosing a random slope in $\F$ for a line with $y$ intercept $\sigma$, 
and then 
evaluating that line at -1 and 1. Both of these values together determine 
the line, but a single one of them does not eliminate any line as a 
possibility. Thus the values at the children of the root \emph{individually} 
contain no information about the value of the root, and in order to decode 
the value at the root, we need both the values at its children. But now, this 
reasoning applies recursively to all the internal nodes, relative to their 
children. Suppose there is a leaf of the tree at which the share is 
missing. Then the share of its parent cannot be decoded because it is 
equally likely to be any element of the $\F$. This propagates up to its 
grandparent, and then its great-grandparent and so on all the way to the root, 
so that the root cannot be decoded.  
Thus, if even one of the shares is missing, the remaining shares provide 
no information about the value of $\sigma$.
\end{proof}

Next, we discuss the tag-and-hash verification scheme used in the protocol. 
This scheme makes it hard for a sender to successfully fool the intended  
recipient of a message by sending a faked message. At the same time, it 
does not give the recipient of the message any information about the message 
prior to receiving it. Such schemes have been used before (see \eg 
~\cite{wegman1981new, rabin1989verifiable, kol2008games}); we include 
the following proposition for completeness.  See Lemma 1 of~\cite{rabin1989verifiable} for the proof.

\begin{proposition}\label{prop:hash}
The verification scheme (Algorithm~\ref{alg:verification}) has the following 
properties:
\begin{enumerate}
\item The verification vector contains no information about the message, \ie, 
the probability of correctly guessing the message given the verification 
vector is the same as the unconditional probability of guessing the message
\item The probability that a faked message will satisfy the verification 
function is $\frac{1}{q-1}$
\end{enumerate}
\end{proposition}

We will now focus our attention on showing that it is a Nash 
equilibrium for all the players to follow our protocol.   
Consider player $j$ and suppose that all other players are committed 
to following the protocol. The next lemma gives a necessary criterion 
for $j$ to have an incentive to cheat.

\begin{lemma}\label{lem:threshold}
If all other players are following the protocol, player $j$ prefers 
to also follow the protocol, unless his probability of successfully 
cheating is at least $\frac{U_j -\Uminus_j}{\Uplus_j -\Uminus_j}$.
\end{lemma}
\begin{proof}
Suppose $j$ is considering deviating from the protocol. We will consider the 
deviation to be successful if either $j$ learns the secret right away, with 
or without being caught, or he does not get caught
and is therefore still in a position to learn it later. The deviation will 
have failed if it is detected, causing the game to end without $j$ learning 
the secret.
Let $p_j$ be the probability that the deviation succeeds.
The maximum utility that $j$ can get is $\Uplus_j$. With probability $1-p_j$, 
the game ends without $j$ learning the secret,
in which case the maximum payoff possible is $\Uminus_j$. Thus a player's 
expected utility from cheating while everyone else follows the 
protocol is at most $p_j\Uplus_j+(1-p_j)\Uminus_j$.

On the other hand, if everyone else follows the protocol, then following 
the protocol guarantees a utility of at least $U_j$. 
Thus the protocol will be a Nash equilibrium if 
\[
U_j > p_j\Uplus_j + (1 - p_j)\Uminus_j 
\]
Rearranging terms, we have a Nash equilibrium if 
\[
p_j < \frac{U_j -\Uminus_j}{\Uplus_j -\Uminus_j}. \qedhere 
\]
\end{proof}

\smallskip

When and how might a player cheat? We note that since players are not 
required to commit to their strategy before starting the game, and
since the progression of the game reveals information, a player may as 
well defer his decision to cheat in a future round until that future round. 
Thus, at any given time, the decision facing the player is whether to 
cheat in the current round. In order to weigh the benefits 
of such a decision, the player needs an estimate of whether the current 
round is likely to be the definitive one. 

As remarked earlier, the purpose of having short and long players is to 
create uncertainty about when the definitive round of the game is, until it 
is too late to gain from this information.  

The players know that $X$ is chosen from a geometric distribution with 
parameter $\beta$. Thus, \emph{a priori} the 
probability that $X$ takes on any particular value is at most $\beta$, 
the most likely being $X=1$, whose probability is exactly $\beta$. 
As the game progresses, players receive partial information about the 
value of $X$; as soon as they receive their inputs they can eliminate all
values of $X$ larger than their input length, if the game did not end on 
the first round, they learn that $X \neq 1$ and so on. Clearly, when a 
player reaches his last block of input, he knows that the current round 
is definitive. The next lemma shows that until that stage, a player's 
estimate that the current round is definitive remains small.

\begin{lemma}\label{lem:guess}
Let $j$ be a player who initially received input for $k>1$ rounds of the game.
and let $1 \le t<k$ be the current round. Then $j$'s estimate of the 
probability that the current round is definitive, conditioned on all the 
information he has learned, is at most $2\beta$. 
\end{lemma}
\begin{proof}
\newcommand{\longp}{\mathcal{L}_j}
\newcommand{\shortp}{\mathcal{L}_j^{\mathrm{C}}}
Let $L_j$ denote the random variable which is the initial input length of 
player $j$. Then we know that 
\begin{equation}\label{eq:Lj}
L_j =\begin{cases} X & \mbox{ if $j$ is a short player} \\
X+Y & \mbox{ if $j$ is a long player}\end{cases}
\end{equation}
Also, let $\longp$ denote the event that $j$ is a long player and 
$\shortp$ the event that $j$ is a short player.  By hypothesis, 
the current round is $t \ge 1$, and player $j$ received an initial input of 
length $k>t$. What information does player $j$ know in round $t$?
\begin{itemize}
\item Since his initial input was of length $k$ he knows that $L_j=k$ and  
$X \le k$ and moreover, that $X=k$ if and only if $\shortp$.
\item Since the game has entered round $t$ he knows that $X\ge t$. 
\item He knows his position $\pi_t(j)$ and the identities of his neighbors 
in round $t$
\item He also has learned $s_t$, $m_{t+1}$ and using the latter to unmask his 
positional data, he knows $\pi_{t+1}(j)$ and the identities of his neighbors
in round $t+1$.  Technically, he learns these just prior to his turn in the 
downstage in round $t$, but this is fine, as we will argue later that 
no player ever has any reason to cheat during the upstage of a round.
\end{itemize}

We note that knowing $s_t$ does not benefit player $j$ in any way as far as 
estimating the probability that $X=t$ goes, since $s_t$ is equally likely to 
be any element of $\S$, independently of $X$. Similarly, knowing the 
identities of his neighbors 
does not affect his estimate, since all other players are equally likely to 
be his neighbors independently of $X$. 

On the other hand, knowing $\pi_t(j)$ and $\pi_{t+i}(j)$ does affect the 
estimate. By construction:
\begin{itemize}
\item In the definitive iteration, short players have odd labels, and 
long players have even labels; and
\item Each player has an odd label in his last round of input 
\end{itemize}
Thus if $\pi_t(j)$ is odd, then player $j$ knows that the current round is not
definitive, since if $X=t$, then $k>t$ implies that $j$ is a long player 
and should have an even label.  In particular, conditioned on everything he knows, 
$\Pr(X=t)=0$. Since $2\beta > 0$ the lemma is proved in this case.

\newcommand{\EEE}{\mathcal{E}_t}
\newcommand{\EEEE}{\mathcal{E}^{b}_{t+1}}

For the remainder of the proof we will assume that $\pi_t(j)$ is even
and denote this event $\EEE$.

Now what about $\pi_{t+1}(j)$? If $k=t+1$, then we know that $\pi_{t+1}(j)$ 
is odd by construction and knowing this contains no additional information 
over knowing $L_j=k$. On the other hand when $k>t+1$, if $\pi_{t+1}(j)$ 
is odd, then player $j$ knows that $X$ cannot be $t+1$ and this affects the 
probability that $X=t$.

Let $b \in \{ 0, 1\}$ be the observed parity of $\pi_{t+1}(j)$ and let 
$\EEEE$ denote the event that the 
parity of $\pi_{t+1}(j)$ is $b$. Note that if $k=t+1$ we must have 
$b=1$.

Let $\Pjt$ be player $j$'s estimate that the current round, 
$t$, is definitive, conditioned on everything he knows. Then
\begin{align*}
\Pjt &= \Pr(X=t \;|\; L_j=k \sand t\le X \le k \sand \EEE \sand \EEEE)\\
&=\frac{\Pr(X=t \sand L_j=k \sand \EEE \sand \EEEE)}
{\Pr(L_j=k \sand t\le X \le k \sand \EEE \sand \EEEE)}\\
&\le \frac{\Pr(X=t \sand L_j=k \sand \EEE \sand \EEEE)}
{\Pr(L_j=k \sand X \in \{t,k\} \sand \EEE \sand \EEEE)}.
\end{align*}
where the inequality follows from the fact that the event $X \in \{t,k\})$ is a subset of the event $t\le X \le k$.

Now, if the current round \emph{is} definitive \ie, $X=t$, then $j$ is not 
a short player, and $L_j=X+Y$. So the event $X=t \sand L_j = k 
\sand \EEE \sand \EEEE$ is the same as the event $\longp \sand
X=t \sand Y=k-t \sand \EEEE$.   Note that $\EEE$ is 
implied  by $\longp \sand X=t$ and can therefore be dropped.

For the denominator, the event $L_j=k \sand X \in\{t, k\} \sand 
\EEE \sand \EEEE$ can be split into the union of disjoint events 
$\shortp \sand X=k \sand \EEE \sand \EEEE$ and 
$\longp \sand X=t \sand Y=k-t \sand \EEEE$. The latter summand 
is the same as the numerator, and loses the $\EEE$ term for the same reason.

Making these substitutions in the above expression, we get 
\[
\Pjt \le \frac{\Pr(\longp \sand X=t \sand Y=k-t \sand \EEEE)}
{\Pr(\shortp \sand X=k \sand \EEE \sand \EEEE)
+\Pr(\longp \sand X=t \sand Y=k-t \sand \EEEE)}
\]
The random variables $X$, $Y$ and the indicator that $j$ is a long player
are independent. Thus the numerator of the above expression becomes
\begin{align*}
\Pr(\longp \sand X&=t \sand Y=k-t \sand \EEEE)\\
&= \Pr(\EEEE|\longp \sand X=t \sand Y=k-t) 
\Pr(\longp)\Pr(X=t)\Pr(Y=k-t)\\
&= \Pr(\EEEE|\longp \sand X=t \sand Y=k-t) 
\frac{\lfloor n/2\rfloor}{n} (1-\beta)^{t-1}\beta (1-\beta)^{k-t-1}\beta\\
&= \Pr(\EEEE|\longp \sand X=t \sand Y=k-t) 
\frac{\lfloor n/2\rfloor}{n} (1-\beta)^{k-2}\beta^2
\end{align*}
Similarly we tackle the first term in the denominator. Since $X$ and all the 
$\pi_i$ are independent and $k\ne t$, then $\pi_t$ and $\pi_X$ 
are independent conditioned on $X=k$. It follows that the events 
$\shortp = \pi_X(j)$ is odd; $\EEE = \pi_t(j)$ is even; and $X=t$ are 
independent.  Thus, we have the following.
\begin{align*}
\Pr(\shortp \sand X&=k\sand \EEE \sand \EEEE)\\
&= \Pr(\EEEE|\shortp \sand X=k \sand \EEE) \Pr(\shortp \sand X=k \sand \EEE)\\
&= \Pr(\EEEE|\shortp \sand X=k \sand \EEE) \Pr(\shortp)\Pr(\EEE)\Pr(X=k)\\
&= \Pr(\EEEE|\shortp \sand X=k \sand \EEE) 
   \frac{\lceil n/2\rceil \lfloor n/2\rfloor}{n^2} (1-\beta)^{k-1}\beta \\
&\ge \Pr(\EEEE|\shortp \sand X=k \sand \EEE) 
\frac{\lfloor n/2\rfloor}{2n} (1-\beta)^{k-1}\beta
\end{align*}

Now, if $k>t+1$ then $\pi_X$ and $\pi_{t+1}$ are independent conditioned 
on $X$ being either $t$ or $k$, and hence, the events 
$\EEEE$ and $\longp \sand X=t \sand Y=k-t$ are independent, as are the 
events
$\EEEE$ and $\shortp \sand X=k \sand \EEE$.
It follows that $\Pr(\EEEE|\longp \sand X=t \sand Y=k-t)$ and 
$\Pr(\EEEE|\shortp \sand X=k \sand \EEE)$ both equal $\Pr(\EEEE)$.
On the other hand if $k=t+1$ then $b=1$ and $\EEEE$ is implied  in both cases
Thus, $\Pr(\EEEE|\longp \sand X=t \sand Y=k-t)$ and 
$\Pr(\EEEE|\shortp \sand X=k \sand \EEE)$ are both 1. Either way, they are 
equal, and since their common value occurs in the numerator as well as in 
both terms in the denominator, it simply cancels out.  

Putting everything together we see that 
{\belowdisplayskip=-14pt
\begin{align*}
\Pjt &\le \frac{\frac{\lfloor n/2\rfloor}{n} (1-\beta)^{k-2}\beta^2}
{\frac{\lfloor n/2\rfloor}{2n} (1-\beta)^{k-1}\beta
+ \frac{\lfloor n/2\rfloor}{n} (1-\beta)^{k-2}\beta^2}\\
&= \frac{\beta}{\frac12(1-\beta) +\beta}\\
&\le 2\beta   
\end{align*} \qedhere}
\end{proof}

We remark that although in the above proof we have bounded player $j$'s
estimate during the down-stage,  a nearly identical proof shows the same bound for the 
up-stage (when $\pi_{t+1}(j)$ is unknown).

We are now ready to prove the main theorem. 

\begin{proof}[Proof of Theorem~\ref{thm:main}]
We will begin by showing that the protocol is a Nash equlibrium 
in which all the players learn the secret.

Suppose all the players follow the protocol. Then every round, during the 
up-stage players send their shares up toward the root where they are decoded, 
and during the down-stage the reconstructed secret and mask are sent back 
toward the leaves.  If the odd players in the round do not drop out at the 
end of the round then play continues into the next round. In the definitive 
round, the real secret is reconstructed at the root and all the even 
labelled players, who are long players learn it first. Once it gets to the 
short players with odd labels, they drop out of the game since they have no 
tags to send any more messages. This signals the end of the game to the long 
players who then realize that the current reconstructed secret is the true one. 
Thus if all the players follow the protocol, everyone learns the secret.

Now suppose all players other than $j$ are following the protocol. We want 
to show that player $j$ prefers following the protocol over 
deviating.  

At the beginning of the game, each player has set their current guess for the 
secret to 0. If no cheating occurred before the current round $t$ then 
during round
$t-1$ all the players set their current guess to $s_{t-1}$. Thus, at the 
beginning of round $t$, all players have the same guess for the secret. 
(Round $t-1$ has been eliminated as the definitive one, but $s_{t-1}$ still has 
probability $1/|\S|$ of being the true secret.) Moreover, since by 
Lemma~\ref{lem:shares} partial information about the shares 
reveals no information about the symbol they encode, throughout 
the up-stage player $j$ has no better guess than $s_{t-1}$ for the secret.
Since $\U < |\S|$, it is strictly better for $j$ not to leave the game in the 
up-stage. If $j$ sends incorrect messages in the up-stage, then even if he 
is not caught(which results in not learning the secret), this deviation will 
cause an incorrect value to be decoded instead of $s_t$.  This results in $j$
not learning the secret if $t$ happened to be the definitive round. Thus, 
$j$ has no incentive to deviate in the up-stage. 

Now what about the down-stage? If $j$ is on his last round of input, then he 
is a short player, and knows that the current round is definitive. At the same 
time, this means that he is an odd-labelled player and by the construction 
of the communication tree he cannot prevent anyone from learning the secret. 
Moreover, even if he successfully fakes a tag in order to convince the unique
long player below him that the game has not ended, that player will detect 
in the next round that all other players have left the game and will 
therefore still output the correct secret. Thus this deviation does not 
change the outcome of the game, namely that all players learn the secret. 
It follows that $j$ does not gain anything 
by this deviation. (Although he also does not lose anything by it.) 
Effectively, if $j$ is a short player on his 
last round of input, it is too late for him to improve his payoff by deviating.

If $j$ is not on his last round of input and is at an odd labelled node then, 
as remarked in the proof of Lemma~\ref{lem:guess}, he \emph{knows} the current 
round is not definitive, so cheating would be equivalent to randomly guessing 
the secret, which is correct with probability only $1/|S|$, This is worse than following the protocol by equation \eqref{eqn:su}.

Now suppose that the current round, $t$, is not $j$'s last round of input 
and $j$ is at an even labelled node.  Note that any spurious messages sent by Player $j$ to players that are not expecting them, will be ignored.  Also, any action involving not sending a message that is expected will be 
detected immediately, only by the involved players at first, but the knowledge will quickly 
propagates to all the players, before the end of the up-stage of the next 
round. Since detection of deviation causes other players to quit immediately, 
effectively such actions amount to player $j$ leaving the game.  
Thus, the possible deviations we need to analyze for player $j$ are:
\begin{itemize}
\item Leave the game, with or without sending fake messages first, and;
\item Send fake messages to one or both children and hope to stay in the 
game by not being caught.
\end{itemize}

Let $\Pjt$
be $j$'s estimate of the probability that the current round is definitive. 
Then by Lemma~\ref{lem:guess},  $\Pjt \le 2\beta$. 

If player $j$ leaves the game and outputs the value $s_t$ then the probability 
that he has output the right value is $\Pjt + (1-\Pjt)/|\S|$, and since he 
has left the game, he has no later opportunity to improve that probability.
By Lemma~\ref{lem:threshold}, in order to discourage this deviation it is 
sufficient if 
\begin{equation}\label{eqn:cheat1}
\Pjt + \frac{(1-\Pjt)}{|\S|} \le \frac{U_j-\Uminus_j}{\Uplus_j-\Uminus_j} 
\end{equation}

Now consider the other kind of deviation. Suppose instead of sending $s_{t}$ to 
his descendents, player $j$ sends a fake value to one or both of them. Let 
$\alpha$ be the probability that he is not caught. By 
Proposition~\ref{prop:hash} we know that $\alpha=\frac{1}{q-1}$ if he sends 
a fake message to only one of his children, and, since the two verification 
functions are chosen independently by the dealer, $\alpha = \frac{1}{(q-1)^2}$
if he fakes both messages.
If he is not caught, and if the current round is definitive, then
player $j$ has learned the true secret and has prevented some of his descendants from 
learning it.   If the current round was not definitive and his deviation 
was not detected, the game continues and since the values $s_i$ are all 
independent, it does not affect the next round.\footnote{This is why player $j$ does 
not try to fake the mask $m_{t+1}$ - an successfully transmitted incorrect 
$m_{t+1}$ will wreak havoc in the next round, since some players will be 
talking to the wrong players.} This means that player $j$ can either revert to 
following the protocol and guarantee learning the secret along with everyone 
else, or he may find further opportunities to cheat.

On the other hand, if the faked message is detected, which happens with probability 
$1-\alpha$, then the game ends right away and player $j$ 
outputs $s_t$. In this case, there is still a $\Pjt$ chance that the current 
round was definitive and an additional $(1-\Pjt)/|\S|$ chance that the 
value $s_t$ was correct despite the current round not having been definitive.
So the probability that the deviation succeeds is 
\[
\alpha + (1-\alpha)\left(\Pjt + \frac{(1-\Pjt)}{|\S|}\right).
\]
As this quantity is bigger when $\alpha= \frac{1}{q-1}$, faking only one 
message dominates faking both messages.
Thus, again by Lemma~\ref{lem:threshold}, to discourage this deviation, it is 
sufficient if
\begin{equation}\label{eqn:cheat2}
\frac{1}{q-1}+ \frac{q-2}{q-1}\left(\Pjt + \frac{(1-\Pjt)}{|\S|}\right)
\le \frac{U_j-\Uminus_j}{\Uplus_j-\Uminus_j}.
\end{equation}
Moreover, note that \eqref{eqn:cheat2} implies \eqref{eqn:cheat1}.
Thus for a Nash equilibrium, it is sufficient to show \eqref{eqn:cheat2}.

For the remainder of the proof, we are going to assume that $n$ is 
sufficiently large, specifically that $n \ge \frac{2\,\U |\S|}{|\S|-\U}$. 
We will discuss the modifications required when 
$3\le n <  \frac{2\,\U |\S|}{|\S|-\U}$ in Section~\ref{sec:smalln}. 

We have so far not specified $\beta$. We do this now.
Let 
\[
\beta = \frac{|\S|-\U}{4\,\U |\S|}
\]
Then $1/\beta =O(1)$ so that the expected number of rounds in the game 
is constant. 

To show \eqref{eqn:cheat2}, recall that $q> n$, so that $q-1 \ge n$.
We have
\begin{align*} 
\frac{1}{q-1}+ \frac{q-2}{q-1}\left(\Pjt + \frac{(1-\Pjt)}{|\S|}\right)
&< \frac{1}{n} + \left(\Pjt + \frac{(1-\Pjt)}{|\S|}\right)
\\&=\frac{1}{n} + \frac{|\S|-1)}{|\S|}\Pjt + \frac1{|\S|}
\\&\le \frac{|\S|-\U}{2\,\U |\S|} + 2\beta  + \frac1{|\S|}
\\&\le \frac{|\S|-\U}{2\,\U |\S|} + 2\frac{|\S|-\U}{4\,\U |\S|}  + 
\frac1{|\S|}
\\&= \frac{|\S|-\U}{\U |\S|} + \frac1{|\S|}
\\&= \frac1{\U}
\end{align*}
as desired.
\smallskip

Finally, we analyze the resource costs of our protocol.  The communication
tree has $2n-1$ nodes. In each round, each player is mapped to one leaf and 
one internal node.  Players only communicate with their neighbors in
the tree. So on each round, during the up-stage each player sends up to three 
messages: one to his parent when he is a leaf; one to his parent when he is 
an non-root internal node; and if he is a child of the root and $n$ is even, 
he has to send an additional message because there are two players at the 
root.

During the down stage each player sends two messages, one each 
to his two children. Thus each player sends at most five messages per round.
Each message consists of four elements of $\F$ (shares of $s_t$, $m_{t+1}$ and 
two tags) each of which is represented as $O(\log n)$, since $n<q\le 2n$
Thus each player send $O(\log n)$ bits per round. 
The expected number of rounds is $1/\beta$ which is constant 
and so each player sends only $O(\log n)$ bits 
during the course of the game. Finally since the tree has depth
$O(\log n)$ the number of rounds is constant (in expectation) and the 
communication is synchronous, it follows that the expected
latency is $O(\log n)$.
\end{proof}

\subsection{Some Remarks}
\subsubsection{The Case of a Small Number of Players}\label{sec:smalln}
When the number of players is a constant greater than $2$, then 
scalability is not an issue, and one might hope to simply use the algorithm of 
Kol and Naor~\cite{kol2008games} by simulating non-simultaneous broadcast 
channels with secure private channels. Unfortunately their algorithm only 
provides an $\epsilon$-Nash equilibrium, since the unique short player 
has a small chance of successfully pretending the game has not ended. 

In our algorithm, $\lceil n/2\rceil$ players are short players. In particular,
even for $n=3$, there are at least two short players, and none of the short 
players can increase their expected payoff by cheating alone. Thus, we 
obtain a Nash equilibrium, provided that we can prove inequality \eqref{eqn:cheat2}. 
The above proof used the fact that $n$ was at least 
$ \frac{2\,\U |\S|}{|\S|-\U}$, so we need a separate argument.

However  since when $n$ is constant, scalability is immediate, we have more leeway to 
choose a larger field to work with.\footnote{The upper bound of $O(n)$ on the field 
size came from the desire to keep $4\log q$, which is the size of an individual 
message, small.} So we can work in a prime field of size $q$ where 
\[
\max\{|\S|, \frac{2\,\U |\S|}{|\S|-\U}\} < q \le 
2 \max\{|\S|, \frac{2\,\U |\S|}{|\S|-\U}\}
\]
and the proof of  \eqref{eqn:cheat2} goes through as before, giving us a Nash 
equilibrium.

\subsubsection{Nash equilibria vs. Strict Nash equilibria}

An $n$-tuple of strategies is called a strict Nash equilibrium
if when all other players are following the prescribed strategy, 
a player unilaterally deviating  achieves a strictly worse expected payoff 
than he would by following the equilibrium strategy. 

Our algorithm fails to be a strict Nash equilibrium, for the following 
reasons:
\begin{itemize}
\item Any player may, at any time, send spurious messages that are not 
part of the protocol, to players that are not his neighbors
in the tree. Such messages will be ignored by their recipients, who are 
following the protocol. 
\item At the end of the definitive round, a short player may try to 
fake a tag and send a message to the long player below him. This may 
go undetected with some small probability, but as noted in the proof,
even in this case, it cannot fool that long player into outputting the 
wrong secret.
\end{itemize}
Our proof shows that our algorithm \emph{does} have the property that 
any player deviating from the protocol in one of the above ways does not increase his 
payoff, and moreover \emph{does not 
affect any other player's payoff either.} In other words, if a player 
deviating unilaterally from our protocol, does so in a manner that 
changes some other player's payoff, then he strictly reduces his own 
expected payoff. 
In this weaker sense, the equilibrium is strict. 

\subsubsection{A Note on Backwards Induction}\label{sec:backind}

The \emph{backwards induction} problem arises when a multi-round
protocol has a last round number that is known to all players.  In
particular, if it is globally known that the last round of the
protocol is $\ell$, then on the $\ell$-th round, there is no longer
any fear or reprisal to persuade a player to follow the protocol.  But
then if no player follows the protocol in the $\ell$-th round, then in
the $(\ell - 1)$-th round, there is no reason for any player to follow
the protocol.  This same logic continues backwards to the very first
round.

The backwards induction problem can occur with protocols that make
cryptographic assumptions, since there will always be some round
number, $\ell$, in which enough time has passed so that even a
computationally bounded player can break the cryptography.  Even
though $\ell$ may be far off in the future, it is globally known that
the protocol will end at round $\ell$, and so by backwards induction,
even in the first round, there is no incentive for a player to follow
the protocol.

As in~\cite{kol2008games}, we protect against backwards induction by
having both long and short players.  As the above analysis shows, if
$N$ is chosen sufficiently large, we can ensure that the probability
of making a correct guess as to when the protocol ends is too small to
enable profitable cheating for any player.  Thus, even when a player
gets to the second to the last element in all his lists, he can not be
very sure that the protocol will end in the next round.  All players
are aware of these probabilities at the beginning of the protocol, and
thus each player knows that no other player will be able to accurately
guess when the protocol ends.

\section{Algorithm for Case of Absent Players}
\label{s:mofn}

In this section we discuss $m$-out-of-$n$ secret sharing where $m
<n$. Here we want a subset of $m$ or more of the players to be able to
reconstruct the secret even when the remaining players are absent. 
However, fewer than $m$ players should \emph{not} be able to
reconstruct the secret on their own.

As discussed previously, it does
not seem possible to design \emph{scalable} algorithms for secret
sharing in the case when either $m$ is much smaller than $n$, or when the
subset of $m$ active players may be chosen in a completely
arbitrary manner.  We now address the situation where (1) $m < n$,
but $m = \Theta(n)$; and (2) the subset of active players does not
depend on the random bits of the dealer.  More precisely, we will
present an algorithm with parameters $\tau$ and $\lambda$ such that when $m > (\tau + \lambda)n$ then with high
probability the algorithm is an Nash equilibrium for the
active players, and all the active players learn the secret. On the
other hand, if $m < (\tau -\lambda)n$,
then with high probability the active players cannot reconstruct the
secret.

We will assume, as in previous work, that the set of active players is
known to all the active players, before the start of the
players' protocol. However, this set is not known to the dealer at the time of 
share creation.   The set of active players may be arbitrary or randomly 
chosen, but it does
not depend on the randomness used by the dealer for the algorithm.

Since the algorithm is a variant of the $n$-out-of-$n$ scheme, for
conciseness, we now describe only the places where the two algorithms
differ.  Let $c$ be a large constant, which we will specify later.
First, the dealer partitions the players into $Q= \frac{n}{c\log n}$
pairwise disjoint groups of $c\log n$ players each.\footnote{It is
  convenient to assume that $c\log n$ is an integer and divides
  $n$, and that the quotient $Q$ is odd. In fact it is always possible 
  to choose $c$ so that $Q$ is odd, and for arbitrary $n$, there will be 
  $Q =\left\lfloor \frac{n}{\lceil c\log n \rceil} \right\rfloor$ groups, each 
  with either $\lceil c\log n \rceil$ or $\lceil c\log n \rceil +1$
  players.}  This is done using a random permutation of the
players, where the first $c\log n$ players in the permutation are the first
group, the next $c\log n$ are the second group and so on. The groups
are labelled $1, 2,\dots, Q$. The communication tree used for the 
algorithm is the labelled $Q$-out-of-$Q$ tree, except that all of the 
nodes are \emph{supernodes}, in that they
correspond to groups of $c\log n$ players instead of single
players. An edge of this communication tree will correspond to all-to-all
communication between the active members of the groups at the endpoints 
of the edge.

As in Algorithm~\ref{alg:Dealer}, the dealer samples $X$ and $Y$ independently
from $G(\beta)$. (Here $\beta = (|\S|-\U)/4\,\U |\S|$ is the same parameter 
as in the $n$-out-of-$n$ algorithm.) For each round $t$, the dealer essentially 
implements the corresponding round of the $Q$-out-of-$Q$ algorithm.  For $t \neq X$, he picks 
random secret $s_t$ for the round.  If $t = X$, $s_X$ is the true secret.  The dealer then picks a random permutation $\pi_t \in S_Q$ to determine which group is 
assigned to which node for that round; and picks a random mask $m_{t+1}$ to encode 
positional data.  A key difference is that now the positional data  constitutes 
not only which node your group is at and which groups are you neighbors, but 
also which $c\log n$ players are in all the relevant groups. For this reason 
we still need to work in a field with size bigger than $n$; a field of size 
merely bigger than $Q$ is not enough.

The dealer then creates $Q$-out-of-$Q$ iterated shares for $s_t$ and $m_{t+1}$. Finally, the values of these 
$Q$-out-of-$Q$ iterated shares at the leaves are further encoded into 
$\tau c\log n$-out-of-$c\log n$ shares via Shaimr's scheme~\cite{shamir:how}. 
These are the shares that will be distributed to the players for round $t$.
This last step is to avoid all the players in a group having the same share,
because if that were the case, then the secret could be decoded by 
$Q= o(m) =o(n)$ players, one from each group. As before, the dealer, creates 
authentication data for all the messages to be sent in the algorithm\footnote{Technically this is not necessary. Since there is all-to-all 
communication between adjacent groups, messages sent by individuals can be 
checked against those sent by other members of the group.};
identifies the short players as the members of the groups at the odd labeled 
nodes in round $X$; truncates shares appropriately; and sends the shares to 
the players.

The players' protocol is very similar to Algorithm~\ref{alg:Player}. The only 
differences are that each active player must send his messages to all active 
players at neighboring nodes, and at the beginning of each round the active 
players at each leaf node must first reconstruct the value at the leaf using 
their $\tau c\log n$-out-of-$c\log n$ Shamir shares. This means that unless 
there are at least $\tau c\log n$ active players in each group, the algorithm 
will fail. Herein lies the reason for the $\delta$ failure probability of the 
algorithm.

Formal descriptions of the protocols for the dealer and players are presented 
as Algorithms~\ref{alg:mofnDealer} and~\ref{alg:mofnPlayer}.

\begin{algorithm}
{\small
\caption{Dealer's Protocol} \label{alg:mofnDealer} 
$\F$ field of size $q$ (to represent messages in the algorithm)
$n$ players with distinct identifiers in $[n] \hookrightarrow \F$,
$\beta \in (0, 1)$: geometric distribution parameter.
$\tau, \lambda, k$ threshold parameters.
\begin{enumerate}
\item Choose $X, Y,$ independently from a geometric distribution with 
parameter $\beta$. Round $X$ is the definitive one. Short players will receive 
full input for $X-1$ rounds and partial input for round $X$. Long players 
will receive full input for $X+Y-1$ rounds and partial input for round $X+Y$. 
\item Let $c =\frac{2(k+1)}{\lambda^2}$ and $Q= \frac{n}{c\log n}$
Choose a random permutation $\pi \in S_n$, and use it to divide players into 
$Q$ groups (numbered $1, 2, \dots, Q$) of size $c\log n$. Use the complete 
binary tree with $Q$ leaves described in Section~\ref{sec:commtree}.
\item For each round $t$ between 1 and $L=X+Y$:
\begin{itemize}
\item If $t<L$, choose a random permutation $\pi_t \in S_Q$. 
If $t=L$ choose a permutation $\pi_L$ which is random subject to the 
constraint that all the long player groups (determined by $\pi_X$) are assigned 
to odd labels under $\pi_L$. For round $t$ player $g$ will be assigned to 
all nodes marked $\pi_t(g)$ in the tree. If $t=1$, $m_1=0$ // (Otherwise 
$m_t$ was set in the previous round)
\item For every group $g$, use $\pi_t$ and $m_t$ to create masked positional 
data for $g$ for round $t$. Positional data consists of the group's position 
and members, neighboring groups and their members.
\item Choose a random mask $m_{t+1} \in \F$ (for the {\bf next} round.)
\item Create shares of $m_{t+1}$ by calling 
\recursiveShares$(\mbox{root}, m_{t+1})$.
\item If $t = X$ $s_t \leftarrow $ true secret.\\
 Otherwise, $s_t \leftarrow $ random element of $\S$
\item Create shares of $s_t$ by calling \recursiveShares$(\mbox{root}, s_t)$.
\item Create $\tau c \log n$-out-of-$c\log n$ Shamir shares of each of the
leaf values of the shares created by \recursiveShares (one for each player 
in the group corresponding to the leaf.)
\item Create authentication data.
\item For each $g$, for each player $j\in g$, $j$'s (full) input $I^j_t$ for 
round $t$  consists of positional data, Shamir shares of recursive shares of 
$m_{t+1}$ and $s_t$ corresponding to node $\pi_t(g)$, and authentication data.
Partial input $\tilde{I}^j_t$ consists of 
all of the above except the authentication tags for sending messages to your 
children (in the down-stage).
\end{itemize}
\item Identify the short players as those players $j$ who are in groups at odd 
numbered nodes in the definitive iteration, \ie, $\pi_X(j)$ is odd.
\item For each short player $j$, send $j$ the list 
$I^j_1,\dots I^j_{X-1}, \tilde{I}^j_X$.
\item For each long player $j$, send $j$ the list 
$I^j_1,\dots I^j_{L-1}, \tilde{I}^j_L$.
\end{enumerate}}
\end{algorithm}

\begin{algorithm}
{\small 
\caption{Protocol for Player $j$} \label{alg:mofnPlayer} 
S=0; M=0 \\
If at any time you receive spurious messages (messages not expected uder the 
protocol), ignore them. \\
On round $t$:
\begin{itemize}
\item[] \upStage:
\begin{enumerate}
\item $m_t = M$
\item Use $m_t$ to unmask and discover your position in the 
tree and the identities of your group members and neighbors for round $t$.
\item (as player at a leaf)
Send your Shamir share to all active members of your group. Receive Shamir 
shares from all active members of your group. If insufficient shares received, 
output $S$ and quit. Otherwise, reconstruct the leaf values of the recursive  
shares of $s_t$ and $m_{t+1}$  and send them along with their tags to all 
active members of the group at your 
parent node in the tree.
\item (as player at an internal node) 
\begin{enumerate}
\item Receive copies of (intermediate) shares of $s_t$ and $m_{t+1}$ 
and tags from active members of the groups at your left and right children 
nodes. Check that correct messages have been sent.
If a fault is detected (missing or incorrect message) output $S$ and quit.
\item For each of $s_t$ and $m_{t+1}$: interpolate a degree 1 polynomial $f$ 
from $(-1, \mbox{ left-share})$ and $(1, \mbox{ right-share})$. 
Evaluate $f(0)$. This is your share.
\item If you are not at the root, send the above reconstructed shares of $s_t$ 
and $m_{t+1}$ to all active members of the group at your parent node. If you 
\emph{are} at the root, these shares are the actual values of $s_t$ and 
$m_{t+1}$.
\end{enumerate}
\end{enumerate}
\item[] \downStage:
\begin{enumerate}
\item If you are at the root, set $S=s_t$ and $M=m_{t+1}$ and send these 
values along with authentification tags to all active members of the groups at
your left and right children nodes.
\item Else
\begin{enumerate}
\item (as a non-root internal node) Receive copies of $s_t$ and $m_{t+1}$ 
and tags from all active members of the group at 
your parent node and check them. If fault detected, 
output $S$ and quit.
\item Set $S=s_t$ and $M=m_{t+1}$.
\item Send $s_t$ and $m_{t+1}$ to all active members of the groups at your 
children nodes.
If you are a short player and have no authentication tags, output $s_t$ and 
quit.
\end{enumerate}
\end{enumerate}
\item[] $t \leftarrow t+1$
\end{itemize}
}
\end{algorithm}

\subsection{Analysis}

In order to show correctness of the algorithm, we need to make two
arguments. The first argument is that the active players are well
distributed among the groups, so that at any 
stage of the players' protocol
sufficiently many shares are available to do the desired
reconstruction. The second argument is that, the protocol is an
Nash equilibrium for the active players. The proof of this
part is essentially identical to the proof that the $n$-out-of-$n$
protocol was a Nash equilibrium for all the players, and
we omit it here. In the remainder of this section, we sketch why
reconstruction is possible with high probability, despite absent
players.

We note that in any round, for the value at an internal node to be 
reconstructed it is necessary that the values of both of its children be
received. Therefore, it is necessary that there is at least one active
player at each internal node.  Since no other reconstruction is to be
done at internal nodes, this is also sufficient. However, note that
the group of $c\log n$ players assigned to an internal node is the
same group as those assigned to some leaf node. Hence if there are no
active players at some internal node, then there is a leaf node at
which there are also no active players. Since at the leaf nodes we
will have a more stringent requirement for how many players need to be
active, it is sufficient to consider the failure of the algorithm at
the leaf nodes.

Now, the value at a leaf node is distributed as $\tau c\log
n$-out-of-$c\log n$ Shamir shares to the $c\log n$ players associated
with that leaf node. Thus, this value can be reconstructed if and only
if there are at least $\tau c\log n$ active players at the leaf
node. Moreover, in order for the protocol to succeed, the values at
all the leaf nodes must be reconstructible.

Recall that the players are assigned to leaf nodes by their group number 
so all leaf node values are reconstructible in a particular round if and 
only if all the groups 
have  at least $\tau c\log n$ active players in that round. Moreover, 
since the same groups are used for all rounds, though not the same 
assignment of them to leaves, and each player is either active or absent for 
the entire game, the following lemma holds.
\begin{lemma}\label{lem:lriffg}
All the leaves are reconstructible throughout the algorithm if and only if 
all the groups have  at least $\tau c\log n$ active players.
\end{lemma}

Now recall that the players are assigned to groups by the 
following random process.  The dealer chooses a random permutation of the $n$
players, and partition the players into $Q= \frac{n}{c\log n}$ 
pairwise disjoint groups by choosing successive contiguous blocks of 
length $c\log n$ in the permutation. Moreover, the choice of which 
players are active is made \emph{independently} of the above process. 
It turns out that in this case, the number of active players in a 
fixed block (\ie \, group) is tightly
concentrated around its mean. A more precise statement is in the
following lemma, whose proof is a simple application of the Azuma-Hoeffding
inequality. We omit the details here.

\begin{lemma}\label{lem:conc}
Let $b_1,b_2, \dots, b_n$ be $n$ bits such that exactly $m$ of them
are $1$ and $n-m$ of them are $0$. Let $\sigma \in S_n$ be a random
permutation of $n$ symbols and $b_{\sigma(1)},b_{\sigma(2)}, \dots,
b_{\sigma(n)}$ be the induced permutation on bits. Fix any
contiguous block $b_{\sigma(i+1)},b_{\sigma(i+2)}, \dots, b_{\sigma(i+
  c\log n)}$ of length $c\log n$, and let random variable $Z$ denote 
the number of bits in the block which are $1$. 
Then $Z$ has expectation $\frac{mc\log n}{n} $ and satisfies the following 
concentration inequalities:
\[
\Pr\left(Z- \frac{mc\log n}{n} > \lambda c\log n\right) \le
e^{-(\lambda^2 c\log n)/2}
\]
and
\[
\Pr\left(Z- \frac{mc\log n}{n} < -\lambda c\log n\right) \le
e^{-(\lambda^2 c\log n)/2}
\]
\end{lemma}

We apply the above concentration inequality to show the following.

\begin{lemma}\label{lem:m}
Let $k\ge 1$, and let $m$ denote the number of active players. If all  
the active players follow the algorithm, 
then with probability at least $1- \frac{1}{n^k}$,
\begin{itemize}
\item If $m \geq (\tau + \lambda)n$, all active players learn the secret.
\item If $m \leq (\tau - \lambda)n$, the secret cannot be reconstructed.
\end{itemize}
\end{lemma}
\begin{proof}
As already remarked in Lemma~\ref{lem:lriffg} all the leaf values of the 
iterated shares of the secret can be reconstructed every round if and only if 
each group has at least $\tau c \log n$ players. Also, by 
Lemma~\ref{lem:shares} the iterated shares can be decoded into the secret
if and only if all the shares at the leaves are available. Thus, provided the  
active players  follow the protocol, the secret is recovered if and only if
each group has at least $\tau c \log n$ players, so that is what we 
need to prove.

Imagine the permutation of players is a bit string such that every bit
corresponds to one player. If a bit is 1 it means the corresponding
player is active and if a bit is 0 the corresponding player is
inactive.  There are two cases.\\

\smallskip
\noindent
Case (1) $m \ge (\tau +\lambda)n$.  Consider a particular group
$g$ and let $Z_g$ denote the number of active players at
$g$. By Lemma~\ref{lem:conc},
\begin{align*}
\Pr(Z_g < \tau c\log n) &= \Pr\left(Z_g -\frac{mc\log n}{n}
< \left(\tau -\frac{m}{n}\right) c\log n\right) 
\\ &\le \Pr\left(
Z_g -\frac{mc\log n}{n} < -\lambda c\log n \right) \\ &\le
e^{-(\lambda^2 c\log n)/2}
\end{align*}
Taking a union bound over all the
$\frac{n}{c\log n}$ groups we see that the probability that the
algorithm fails to recover the secret, which happens only when  
\emph{some} group does not have enough active players, is at most 
$ne^{-(\lambda^2 c\log n)/2} =
n^{1-\lambda^2c/2}$. Setting $c = \frac{2(k+1)}{\lambda^2}$ we see
that the probability that the algorithm fails is at most $1/n^k$.\\

\smallskip
\noindent
Case (2) $m \le (\tau -\lambda)n$.  Once again consider a particular group
$g$ and let $Z_g$ denote the number of active players at $g$.
 By Lemma~\ref{lem:conc},
\begin{align*}
\Pr(Z_g > \tau c\log n) &= \Pr\left(Z_g -\frac{mc\log n}{n}
> \left(\tau -\frac{m}{n}\right) c\log n\right) \\ &\le \Pr\left(
Z_g -\frac{mc\log n}{n} > \lambda c\log n \right) \\ &\le
e^{-(\lambda^2 c\log n)/2}
\end{align*}
Therefore with probability at least $1- e^{-(\lambda^2 c\log n)/2} =
1- 1/n^{k+1}$, $g$ does not have enough active players. Since the failure of 
a single group to have enough active players is sufficient to break the
protocol completely (by Lemma~\ref{lem:shares}) it follows that with
probability at least $1-1/n^{k+1} \ge 1-1/n^k$ the algorithm fails to recover 
the secret.
\end{proof}

Why do players follow the protocol? As in the $n$-out-of-$n$ case, we can 
argue that a player looking at his remaining input has a very low estimate 
of the current round being definitive, unless he is actually on his last round. 
In his last round of input, there is nothing he can do to prevent others 
from learning the secret that would not also prevent himself from learning 
it. We omit the details, which are essentially the same as in the proofs of 
Lemma~\ref{lem:guess} and Theorem~\ref{thm:main}. We conclude with

\begin{proof}[Proof of Theorem~\ref{thm:mofn}]
Lemma~\ref{lem:m} shows that with high probability, when the fraction of 
active players is higher than $\tau + \lambda$, all active players following 
the algorithm leads to their all learning the secret, while if the 
fraction is less than $\tau-\lambda$ then nobody learns it. We have also 
remarked that no player can do better by  deviating from the protocol, so 
that it is a Nash equilibrum. 
The communication tree has depth $\log Q = \log n -\log\log n -\log c 
= O(\log n)$. 
On each round, each player sends messages to $\Theta(\log n)$ others, and 
the messages  themselves are $\log q  = O(\log n)$ bits. The algorithm runs 
for $1/\beta =O(1)$ steps in expectation. It follows that the algorithm has 
latency $O(\log n)$ and each player sends $O(\log^2 n)$ bits in expectation.
\end{proof}

\section{Conclusion}
\label{s:conclusion}

We have presented \emph{scalable} mechanisms for rational secret sharing problems.  Our algorithms are scalable in the sense that the number of bits sent by each player is $O(\log n)$ and the latency is at most logarithmic in the number of players.  For $n$-out-of-$n$ rational secret sharing, we give a scalable algorithm that is a Nash equilibrium to solve this problem.  For $m$-out-of-$n$ rational secret sharing where (1) $m = \Theta(n)$; and (2) the set of active players is chosen independently of the random bits of the dealer, we give a scalable algorithm with threshold parameter $\tau$ that is a Nash equilibrium and ensures that for any fixed, positive $\lambda$ that if (1) at least a $m/n > \tau + \lambda$ fraction of the players are active, all players will learn the secret; and (2) if fewer than a $\tau -\lambda$ fraction of the players are active, then the secret can not be recovered. 

Several open problems remain.  First, while our algorithms lead to a $\Theta(n)$ multiplicative reduction in communication costs for rational secure multiparty computation (SMPC), the overall bandwidth for this problem is still very high. We ask: Can we design scalable algorithms for rational SMPC?    This is related to our second open problem which is: Can we design scalable algorithms for simulating a class of well-motivated mediators?  In some sense, this problem may be harder than the SMPC problem, since some types of mediators offer different advice to different players.  In other ways, the problem is easier: a simple global coin toss is an effective mediator for many games.  A final important problem is: Can we design coalition-resistant scalable algorithms for rational secret sharing?

\subsection*{Acknowledgments}

We are grateful to Tom Hayes and Jonathan Katz for useful discussions.

\bibliographystyle{plain}
\bibliography{security}
\end{document}